\documentclass[conference]{IEEEtran}
\ifCLASSINFOpdf
\else
\fi
\usepackage[cmex10]{amsmath}
\usepackage{amssymb,amsfonts}

\newcommand{\dirac}[1]{| #1 \rangle}

\newcommand{\li}[1]{\overline{#1}}

\newtheorem{theorem}{Theorem}
\newtheorem{lemma}{Lemma}
\newtheorem{corollary}{Corollary}
\newtheorem{proposition}{Proposition}
\newtheorem{definition}{Definition}
\begin{document}
%
\title{A NEW APPROACH TO CODEWORD STABILIZED QUANTUM CODES USING THE ALGEBRAIC STRUCTURE OF MODULES}

\author{\IEEEauthorblockN{DOUGLAS FREDERICO GUIMAR\~AES SANTIAGO\\ and GERALDO SAMUEL SENA OTONI}
\IEEEauthorblockA{Instituto de Ci\^encia e Tecnologia - ICT\\
Diamantina, Brasil\\
Email: douglas.santiago@ict.ufvjm.edu.br}}


%


\maketitle

\begin{abstract}
In this work, we study the Codeword Stabilized Quantum Codes (CWS codes) a generalization of the stabilizers quantum codes using a new approach, the algebraic structure of modules, a generalization of linear spaces. We show then a new result that relates CWS codes with stabilizer codes generalizing results in the literature.
\end{abstract}


%
\IEEEpeerreviewmaketitle

\section{Introduction}

With the development of quantum computing, as well as in classical computing, the emergence of mechanisms to detect and correct errors should be implemented, then follows the need of the theory of quantum error correction codes (\cite{ Nielsen2000}, \cite{Knill1997}, \cite{gaitan2008quantum}, \cite{Aharonov97fault-tolerantquantum}, \cite{Calderbank1996},\cite{Bennett1996}, \cite{Steane1996}, \cite{Gottesman1996}). Protection against quantum error involves different challenges than protecting against classic mistakes, but despite this, much of the classical theory of error-correcting codes can be harnessed for quantum codes.

A quantum code is a subspace of a Hilbert space and is usually represented by the parameters $ ( ( n , K , e) )_d $. The parameter $ d $ is the amount of quantum levels being considered, e.g, the number of linearly independent states a single qudit can present. The parameter $ n $ is the dimension of the larger Hilbert space, $ K $ is the dimension of the code. The parameter $e$ is the number of qudits that the code can detect.

A class of quantum codes much explored in the literature is the class of stabilizer codes (\cite{Gottesman1997}, \cite{Calderbank97quantumerror} ). In these, the subspace which defines the code is the intersection of the subspaces associated with the eigenvalue $1$ of a set of operators that form a subgroup of the Pauli group. This group is called the stabilizer group $ S $.

In a CWS code (Codeword Stabilized Quantum Codes) with parameters $( ( n , K , e) )_d$, the stabilizer group stabilizes a single quantum state (up global phase) and the basis elements are constructed by applying distinct Pauli Operators in the stabilizer state (\cite{Chen2008}, \cite{Chuang2009}, \cite{Cross2009}, \cite{Hu2008}, \cite{Yu2008}, \cite{Yu2007},  \cite{Yu2009}). The CWS codes are a generalization of the stabilizers codes, since it has been proved that every stabilizer code can be seen as a CWS code. Conversely, it was also proved that a CWS code satisfying certain conditions is actually a stabilizer code. There are several results in the literature about the CWS codes and one of these allow us to construct quantum CWS codes with the higher possible parameter $K$ with parameters $ n $ and $ e $ fixed. The problem of constructing good  CWS (with parameter $ K $ large) becomes then the problem of constructing good classical codes to correct a particular set of errors. The theory of CWS codes then presents a method to create new quantum codes (stabilizer or not) based on classical codes. 

 This work present a new approach in the study the CWS codes by the algebraic structure of modules and the generalization of the concept of parity check matrix. We also present Theorem~\ref{LinearisSta}, that generalizes results found in the literature and helps to determine when a CWS code is a stabilizer code.

In the second section, we explain in more details the structure of CWS codes, based mainly in \cite{Chen2008}. In the third section we introduce and generalize the notion of parity matrix. In the fourth section, we present some necessary results on the theory of modules. In the fifth section we prove some known results about stabilizer spaces using the concept of parity matrix as done in \cite{Gheorghiu2014505}, but using also the structure of modules. In the sixty section we present our main result (Theorem~\ref{LinearisSta}). The Corollaries~\ref{col1} and \ref{col2} concerning this theorem represent well known results in the literature, although we also have not found a prove on qudits for these results.

 

\section{Structure of CWS codes}
For a qudit, the Pauli group $ \mathcal{G}_d^1$ is generated by $X$, $Z$, where the commute relation is given by $$ZX = q_dXZ$$ and $q_d=e^{i\frac{2\pi}{d}}$. Note that setting this way, for a qubit ($ d = 2$) the Pauli group $\mathcal{G}_2^1$, which in the binary case we also represent by $\mathcal{G}$, is given by 
\begin{equation}
\mathcal{G}_2^1=\{ I, -I, Z , -Z , X, -X, ZX, -ZX \} .
\end{equation}
There is a representation of $\mathcal{G}_d^1$ and a basis $\{\dirac{k}\}_{k=0}^{d-1}$ such that
\begin{equation}
Z\dirac{k}=q_d^k\dirac{k},\;\;X\dirac{k}=\dirac{k+1}, \textrm{para todo } k \in Z_{d}.
\end{equation}
It follows that $Z^jX^k=q_d^{jk}X^kZ^j$ and general relation (\cite{Ketkar2006}) is given by

\begin{eqnarray}\label{comut1}
&(q_d^{i_1}Z^{j_1}X^{k_1})(q_d^{i_2}Z^{j_2}X^{k_2}) & \\
&= q_d^{j_1k_2-k_1j_2}(q_d^{i_2}Z^{j_2}X^{k_2})(q_d^{i_1}Z^{j_1}X^{k_1})&\nonumber
\end{eqnarray}

Considering these commute relations, an element of the Pauli group $\mathcal{G}_d^n=\underbrace{\mathcal{G}_d^1}_{1}\otimes\ldots,\otimes \underbrace{\mathcal{G}_d^1}_{n}$ may be written as
$$\alpha Z^{\mathbf{V}}X^{\mathbf{U}}$$ where $\alpha=q_d^{k}$ where $\mathbf{V}$ e $\mathbf{U}$ represent vectors in $\mathbb{Z}_d^n$ indicating the power of $ Z $ and $ X $ on each qudit respectively. Extending the commute relation \ref{comut1} we have

\begin{eqnarray}\label{comut2}
 &(Z^{\mathbf{U}_1}X^{\mathbf{U}_2})(Z^{\mathbf{V}_1}X^{\mathbf{V}_2})& \\
 & =q_d^{\langle \mathbf{U}_1,\mathbf{V}_2\rangle - \langle \mathbf{U}_2,\mathbf{V}_1\rangle}(Z^{\mathbf{V}_1}X^{\mathbf{V}_2})(Z^{\mathbf{U}_1}X^{\mathbf{U}_2}) & \nonumber
\end{eqnarray}

where $ \langle . , . \rangle $ denotes the canonical inner product restricted to $ \mathbb{Z}_d^n$, which is not necessarily a linear space. If $ d $ is prime $Z_d^n$ is a linear space, and in this case,  $Z_d$ is a field, otherwise $ Z_d^n $ has the structure of a module.

Disregarding global phase, We represent one Pauli operator $E=\alpha Z^{\mathbf{U}_1}X^{\mathbf{U}_2}$ as expanded vector in $ \mathbb{Z}_d^{2n} $. This is done by applying the function $R$ defined as follows:

\vspace*{12pt}
\noindent
{\bf Definition~1:}\label{funcaoR}
Let $\mathcal{G}_d^n$ be the qudit Pauli Group with $n$ entries and the $\mathbb{Z}_d$-module, $\mathbb{Z}_d^{2n}$ . The function $R$ is defined as 
\begin{align*}
R:\mathcal{G}_d^n &\rightarrow  \mathbb{Z}_d^{2n}\\
 \alpha Z^{\mathbf{U}_1}X^{\mathbf{U}_2} &\mapsto  (\mathbf{U}_1|\mathbf{U}_2).
\end{align*}
\vspace*{12pt}
\noindent

Clearly the function $ R $ is well defined, is surjective but not injective, since the information contained in the phase $ \alpha $ is lost. the function $ R $ is also a group homomorphism, e.g, $R(g_1g_2)=R(g_1)+R(g_2)$ and $R(g^{\dagger})=-R(g)$. Using the representation of the elements of $ \mathcal{G}_d^n $ given by the function $ R $, we can determine the phase that appears in the general commute relation (\ref{comut2}) through the operator of dimension  $2n\times 2n$ times defined by

\begin{equation}\label{lambda}\Lambda=\left[\begin{array}{cc}
0 & I\\ -I & 0
\end{array}\right],
\end{equation}
where $ 0 $ and $ I $ refers to zero and identity submatrices, respectively of dimension $n\times n$. Using this operator, we note that any two Pauli operators in $ P_1 $ and $ P_2 $ obey the commute relation
\begin{equation}\label{comut3}
P_1P_2=q_{d}^{R(P_1)\Lambda R^T(P_2)}P_2P_1.\footnote{$R(P)$ is a line vector $R^T(P)$ is the transposed line vector.}
\end{equation}

 The operation $R(P_1)\Lambda R^T(P_2)$ is known as symplectic product (\cite{Ketkar2006}).

We can, according to \cite{Chen2008}, construct a CWS code by two sets:
\begin{enumerate}
\item An Abelian group $S=\langle s_1,\ldots,s_r\rangle$ of order $ | S | = d^n $ not containing multiples of the identity except the identity itself ( this group stabilizes, disregarding global phase, a single state $ \dirac{\psi} \in \mathcal{H}_d^n$ );
\item A set $ W = \{w_i\}_{i=1}^{K}$ where $\{w_i\} $ are Pauli operators such that $ \beta=\{w_i\dirac{\psi}\}$ represents the code base.
\end{enumerate}

Moreover, we can verify that these conditions guarantee that all operators $ S $ are simultaneously diagonalizable, e.g, there exists a common basis of eigenvectors to all operators of $ S $. Regardless of $S$ stabilize a single state or not, let's call this group \textit{Stabilizer Group}.

\section{Parity Check Matrix}

We can  represent a collection of Pauli operators through a matrix as the way the theory of classical error correction codes does ( \cite{MacWillians1977}, \cite{vermani1996}). We will call this matrix by \textit{Parity Check Matrix}, or simply \textit{Parity Matrix}. This matrix is already used in the formalism of stabilizer codes (\cite{Nielsen2000}) with the generators of the stabilizer group, but here we define in general, for any set of Pauli operators.

\begin{definition}\label{MatrizV} Given a collection of Pauli operators  $\mathcal{C}=\{ p_1,\ldots, p_r\}$ in $ \mathcal{G}_d^n $, we call \textbf{parity check matrix} of $ \mathcal{C}$, $R(\mathcal{C})$, the matrix of size $ { r \times 2n} $ where each row of the matrix $R$ is the vector $ ( p_i ) $. 
\end{definition}

Given a stabilizer group $ S $ with generators $ \mathbb{S}=\{ s_1, \ldots, s_r \} $, $ R ( \mathbb{S} ) $ will be the parity check matrix of size $ r \times 2n $ about the collection of generators of $ S $. The $ \mathbb {Z}_d $-module generated by the rows of the parity check matrix over $ \mathbb{S}$ will be denoted by $ \langle R ( \mathbb { S } ) \rangle $. It is easy to verify that $ | S | = \# \langle R (\mathbb{S}) \rangle $, where the symbol $\#$ denotes the cardinality of the set. \footnote{The concept of parity check matrix will also be used on the collection of \textit{codewords} $ W = \{w_i\}_{ i = 1 }^{K}$, generating a parity check matrix $ R(W)$ of size $ K \times 2n $.}

It is useful at this point to enunciate the most important theorem for CWS codes using the parity check matrix definition. This theorem allow us to create quantum CWS  codes looking for classical codes \cite{Chuang2009}, \cite{Cross2009}

\begin{theorem}\label{QuaCla}
Let $ \mathcal{Q} $ be a CWS code with stabilizer generators $\mathbb{S}=\{s_1,\ldots,s_r\}$,  \textit{codewords} $ W = \{ w_i \}_{ i = 1 }^{ K } $, $ w_1 = I $, $ \epsilon = \{ E \} $ a set of Pauli errors and let $Cl_S$ be the function
\begin{equation}\label{funcaocls}
Cl_S(P)=R(\mathbb{S})\Lambda R^T(P).
\end{equation}
Then the code $ \mathcal{Q} $ detects errors in $ \epsilon $ if and only if $ Cl_S(W) $ detects errors in $ Cl_S( \epsilon ) $ and moreover, if $ Cl_S ( E) = 0$  then \begin{equation}\label{condicaoQ} Ew_i=w_iE\end{equation} for all $ i $.

\end{theorem}

\section{Modules}

 The algebraic structure of modules can be seen in \cite{Ricou2004}. In this work, we use repeatedly that, given a homomorphism from $ \mathbb{Z}_d $-modules represented by a matrix $ \mathcal{T}$, the cardinality of the module generated by the rows of $ \mathcal{T} $, which we denote by $ \langle \mathcal{T} \rangle $ is equal to the cardinality of the module generated by the columns of $ \mathcal{T} $, which can be represented by the module $ Im(\mathcal{T} ) $. We will refer to these modules as row-modules and column-modules, respectively.

The isomorphisms theorems for modules \cite{Ricou2004} will be used frequently in the proofs of this work   

\begin{theorem}\label{isomorphismstheorems}
Let $A$ be a ring.
\begin{enumerate}
\item If $\phi: M_1 \rightarrow M_2$ is a $A$-module homomorphism, then there is a isomorphism:
$$Im(\phi)\simeq \frac{M_1}{Ker(\phi)}.$$ 
\item If $N_1,N_2$ are submodules of a $A$-module $M$, there is a isomorphism:
$$\frac{N_1+N_2}{N_2}\simeq \frac{N_1}{N_1 \cap N_2}.$$
\item If $N,P$ are submodules of a $A$-module $M$ and $P \subset N\subset M$ so $P$ is an submodule of $N$ and there is a isomorphism:
$$M/N\simeq \frac{M/P}{N/P}.$$
\end{enumerate}
\end{theorem}


 The definitions of elementary operations used in this work are\footnote{It's important to remark that all elementary operations are made in $\mathbb{Z}_ {d}$}
 
\begin{definition}\label{elementaryoperations}
The elementary operations are given by:
\begin{enumerate}
\item Exchange two columns/rows ($C_i \leftrightarrow C_j$ or $R_i \leftrightarrow R_j$).
\item Add a column/row with the multiple in $\mathbb{Z}_d$ of other column/row
 ($ C_ {i} \rightarrow C_ {i} + \beta C_ {j}$ or $ R_ {i} \rightarrow R_ {i} + \beta R_ {j}$).
\end{enumerate}
\end{definition}

Given a matrix $\mathcal{T}$ with entries in $\mathbb{Z}_d$, we will first prove that elementary operations in their rows or column do not change the cardinality of the row and column modules. For this we assume a matrix $ \mathcal{T} = [C_ {1} , C_ {2} , ... , C_ {n}] $, where
$ C_ {i} $ is a column vector in $ \mathbb{Z}_ {d} ^ {k} $.

Then the column module is:

\[
Im(\mathcal{T})=\{X_{1}C_{1}+...+X_{n}C_{n}/X_{i}\in \mathbb{Z}_{d}\},
\]

Clearly exchange between two columns do not change the cardinality of  $Im(\mathcal{T})$. Neither the operation $ C_ {i} \rightarrow C_ {i} + \beta C_ {j}$, as we will see below. Without loss of generality, assume the operation using the first and second column, then

\[
Im(\mathcal{T}')=\{X'_{1}(C_{1}+\beta C_{2})+...+X'_{n}C_{n}/X'_{i}\in Z_{d}\},
\]
But $Im(\mathcal{T}')\subset Im(\mathcal{T})$. To see this, make $ X'_ {1} = X_{1} $, $ X'_{2} =(X_{2} - \beta X_{1}) $, $X'_{3}=X_{3}$,
...,$X'_{n}=X_{n}$.

We also have $Im(T)\subset Im(\mathcal{T}')$. To see this, make $ X_ {1} = X'_{1} $, $ X'_{2} =(X_{2} + \beta X_{1}) $, $X_{3}=X'_{3}$,
...,$X_{n}=X'_{n}$.

Elementary operations with the lines also do not affect the column module. To see this, consider the matrix $ \mathcal{T} $ as
$\mathcal{T}=\left[\begin{array}{c}
R_{1}\\
R_{2}\\
\vdots\\
R_{k}
\end{array}\right]$ where each $ R_{i} $ is a row vector in $ Z_{d}^{n} $ and $R_{ij}$ your components. The $Kernel$
of $ \mathcal{T} $ is:

\begin{eqnarray}
 & Ker(\mathcal{T})=\{[X_{1},X_{2},...,X_{k}] & \\
 & /R_{i1}X_{1}+R_{i2}X_{2}+...+R_{in}X_{n}=0\;\forall i\in (1\ldots k)\} & \nonumber
\end{eqnarray}

Clearly exchanging the lines do not change the cardinality of $Ker(\mathcal{T})$. With a analogous proof for the column module we see that   $R_{i}\rightarrow R_{i}+\beta R_{j}$ also do not change it. How $Ker(\mathcal{T})$ do not change with elementary operations with lines, by the first isomorphism theorem \ref{isomorphismstheorems}, the cardinality of $Im(\mathcal{T})$ also do not change which means that the number of elements of the column module also do not change.

To show that the elementary operations do not change the
line module we can make a procedure analogous to the previous one using
the transposed matrix $ \mathcal{T}^{T} $.

The next step is to make elementary operations until $T$ get a form in wich we can see that both cardinalities are equal. We will need the lemma:

\begin{lemma}\label{divisao}
Let $a$ and $b$ be integers, $0\leq a,b\leq d-1$ and $a\neq 0$, so there exist $\overline{q},\overline{r}\in \mathbb{Z}_d$ satisfying $0\leq r <a$ and 
$$\overline{r}=\overline{a}\,\overline{q}+\overline{b}$$
\end{lemma}
\begin{proof}
Using the Euclid's division algorithm,there exist $q'$ and $0\leq r<d-1$ satisfying
$$b=a\,q'+r$$
then
\begin{align*}
\overline{b}=\overline{a}\,\overline{q'}+\overline{r}
\end{align*}
inthis case, how $0\leq b<d$, we have also $0\leq q' <d$. Take $q=d-q'$ and we have
$$\overline{r}=\overline{b}+\overline{a}\,\overline{q}$$
\end{proof}

The next proposition allow us to obtain an equivalent to gaussian elimination through elementary operations.

\begin{proposition}\label{pelement}
Through elementary operations (with the column elements) we can transform $V=(\li{v_1},\ldots, \li{v_n})$ $\in \mathbb{Z}_d^n$ with $0\leq v_i <d$ and at least one not null entry, in $V'=(\li{a},\li{0},\ldots,\li{0})$ wih only the first entry $\li{a}$ assuming a not null entry
\end{proposition}
\begin{proof}
Repeat the process:
\begin{enumerate}
\item Let $v_j$ be one entry with the least not null absolute value. Exchange the $j$ entry with the first. Then rename the entries to $V=(\li{a_1},\ldots, \li{a_n})$.
\item For each $j\neq 1$ use Lemma~\ref{divisao} to obtain in the $j$ entry,
$$\li{r_j}=\li{a_j}+\li{q_j}\,\li{a_1}\,\;\;\textrm{onde}\;\;0\leq r_j< a_1.$$
\item Repeat procedures 1 and 2 until get the result.
\end{enumerate}
\end{proof}

In the  next proposition we get the statement about the equality of the cardinalities of the row and columns modules.
\begin{proposition}\label{cardinalidades}
Let $\mathcal{T}_{(m\times n)}$ an matrix with entries in $\mathbb{Z}_d$ representing an $\mathbb{Z}_d$-module homomorphism, then  the cardinalities of the row and columns modules are equal, $\#Im(\mathcal{T})=\#\langle \mathcal{T}\rangle$.
\end{proposition}
\begin{proof}
How elementary operations do not change the cardinalities of the row and columns modules, just follow the procedure:
\begin{enumerate}
\item Consider together all the first row an first column values of $\mathcal{T}$. Take the least of them and through exchange elementary operations, put it on the $(1,1)$ position. 
\item Still considering together all the first row an first column values of $\mathcal{T}$, make how Proposition~\ref{pelement}. After this procedure, we make null all the first row an first column values of $\mathcal{T}$ but the $(1,1)$ position.
\end{enumerate}
repeating this procedure to the others rows and columns, we obtain an matrix $T'$ in wich only the $(i,i)$ positions with $i\in $  $1,\ldots \textrm{least}(n,m)$ may assume not null values. Clearly this matrix satisfies $\#Im(\mathcal{T}')=\#\langle \mathcal{T}'\rangle$ How elementary operations do not change the cardinalities of the row and columns modules, we get the statement. 
\end{proof}

\section{Stabilizer Spaces}\label{structure}

A first question that arises is if the fact that the stabilizer group $ S $ be abelian, not containing multiple of the identity but the identity itself and $ | S | = d^n $ are necessary and sufficient conditions to $ S $ stabilize a single phase state\footnote{the phrase \textit{stabilize a single phase state} should be considered always disregarding a global phase} $ \dirac{ \psi } $. The answer to this question is positive. For the binary case the result is demonstrated in \cite{Nielsen2000} and makes use of the parity check matrix $ R ( \mathbb{S} ) $. We can extend this statement for the case $ d $ prime. We also can prove that the result holds for any $ d $ using ideas contained in \cite{Arvind2002} and \cite{Gheorghiu2014505}. Here,  we chose to make a new approach, similar to that made for qubits, using the parity check matrix $ R ( \mathbb{S} ) $ and the interpretation of the matrix $ R ( \mathbb{S}) \Lambda $ as a homomorphism between $ \mathbb{Z}_d$-modules.

\begin{lemma}\label{lematamanho1}
Let $ S = \langle s_1 , \ldots , s_r \rangle $ be an abelian subgroup of the Pauli group $ \mathcal{G}_d^ n$ not containing multiple of the identity other than the identity itself. If $ | S | < d^n $, then we can add an element $ P \in \mathcal{G}_d^n \backslash S $ such that $ \overline{S} = \langle s_1 , \ldots , s_r , P \rangle $ is still an abelian group not containing multiple of the identity other than the identity itself.
\end{lemma}
\begin{proof}
The $ \Lambda $ operator  does not change the cardinality of the row-module, so we have $ | S | = \# \langle R ( \mathbb{S} ) \rangle = \# \langle R ( \mathbb{S} ) \Lambda \rangle < d^n $, and as the cardinality of the row-module is equal to the column-module, it follows that $ \# Im(R(\mathbb{S})\Lambda ) < d^n $. By the first isomorphism theorem for modules,
$$Im(R(\mathbb{S})\Lambda)\simeq \frac{\mathbb{Z}_{d}^{2n}}{Ker(R(\mathbb{S})\Lambda)}$$ 
 which means that $\#Ker(R(\mathbb{S})\Lambda)>d^n$, so there is $ \overline{P} \in \mathcal{G}_d^n \backslash S $ that commutes with all elements of $ S $. Let $o$ Be the first natural number such that $ \overline{P}^o = \alpha I $. Take $ \beta \in \mathbb{C}$ such that $ \beta^o\alpha = $ 1. Therefore $ P = \beta \overline{P}$ commutes with all elements of $ S $ and $ P^o = I $, so $ \overline{S} = \langle s_1, \ldots, s_r, P \rangle $ is an abelian group not containing multiple of the identity other than the identity itself.
\end{proof}
 
\begin{lemma}\label{lematamanho2}
Let $ S = \langle s_1, \ldots, s_r \rangle $ be an abelian subgroup of the Pauli group $ \mathcal{G}_d^n $ not containing multiple of the identity other than the identity itself. Then $ | S | \leq d^n $.
\end{lemma}
\begin{proof}
The demonstration follows similar to the proof of Lemma~\ref{lematamanho1}. Suppose that $ | S | > d^n $. By the first isomorphism theorem for modules we have
$$Im(R(\mathbb{S})\Lambda)\simeq \frac{\mathbb{Z}_{d}^{2n}}{Ker(R(\mathbb{S})\Lambda)},$$ 
where
$\#Ker(R(\mathbb{S})\Lambda)<d^n$, which is a contradiction because all elements of $ \langle R(\mathbb{S} ) \rangle $ belong to $ Ker(R(\mathbb{S} ) \Lambda ) $.
\end{proof}

The next theorem relates the order of the stabilizer group with the dimension of the stabilized quantum code $ \mathcal{Q} $. To understand it, we will start now an argument that will culminate with the theorem.

All Pauli operator $ P $ is an isomorphism between linear spaces, so if $ \mathcal{Q} $ is a quantum code, $ P \mathcal{Q} $ is a quantum code with the same dimension of $ \mathcal{Q} $. If $ \mathcal{Q} $ is stabilized by $ S = \langle s_1, \ldots, s_r \rangle $ then according to the formalism of stabilizers, $ P \mathcal{Q} $ is stabilized by $ S'= PSP^{\dagger } $. The generators of $ S' $ are
$$ S' = \langle q_{d}^{d- \alpha_1  }s_1, \ldots, q_{d}^{d-\alpha_r }s_r \rangle $$
where the vector $ (d-\alpha_1 \ldots, d-\alpha_r ) $ is obtained using the equation~\ref{comut3} according to the following operation
$$ R(\mathbb{S}) \Lambda R^{T}(P^{\dagger}). $$
If $ \mathcal{Q} $ is stabilized by $ S $, then $ \mathcal{Q} $ is the eigenspace associated to the eigenvalue 1 of each operator $ \mathbb{S} = \{ S_i \} $, then $ P \mathcal{Q} $ is the eigenspace associated with the eigenvalues ​​$ \{q_{d}^{\alpha_i} \} $ of each operator on $ \mathbb{S} $. Considering then the homomorphism between modules represented by the matrix
$$ R ( \mathbb{S} ) \Lambda $$
we have that every element $ \mathbf{x} $ on the image of this homomorphism, $ \mathbf{x} \in Im(R(\mathbb{S}) \Lambda ) $, represents a distinct subspace of $\mathcal{H}_{d}^{n} $. We know that they are distinct because subspaces associated to distinct eigenvalues ​​has only trivial intersection.

By lemmas~\ref{lematamanho1} and \ref{lematamanho2} we know that we can complete the stabilizer group $S=\langle s_1, \ldots, s_r \rangle $ such that $ S' = \langle s_1 , \ldots , s_r , P_1 , \ldots , P_M \rangle $ is a stabilizer group and has order $ | S '| = d^n $. Let $ \mathbb{S}'= \{ s_1 , \ldots , s_r , P_1 , \ldots , P_M \} $. Since $ | S '| = \#\langle R ( \mathbb{S}' ) \rangle = \# \langle R ( \mathbb{S}') \Lambda \rangle $ and the cardinality of the column-module is equal to the row-module, so
$$ \# Im(R(\mathbb{S}')\Lambda) = d^n. $$
As each element $ \mathbf{x} = ( \alpha_1 , \ldots , \alpha_r , \beta_1 , \ldots , \beta_m ) $ of $ Im(R(\mathbb{S}')\Lambda) $ represents a distinct subspace of same dimension, and the dimension of the whole space is $dim(\mathcal{H}_{d}^{n}) = d^n $, it follows that every subspace $ V_\mathbf{x} $ stabilized by $\mathbb{S}'= \langle q_{d}^{ d-\alpha_1  }s_1, \ldots, q_{d}^{d-\alpha_r }s_r, q_{d}^{d-\beta_1 } P_1, \ldots , q_{d}^{d-\beta_m  }P_M \rangle $ has dimension 1 and the union of these ones covers the whole $ \mathcal{H}_{d}^{n} $. Since each $ V_\mathbf{x}$ is a subspace of the space stabilized by $\mathbb{S}=\langle q_{d}^{d-\alpha_1}s_1,\ldots,q_{d}^{d-\alpha_r}s_r\rangle$, they also cover $ \mathcal{H}_{d}^{n} $, and the same has trivial intersection, so the subspace $ \mathcal{Q} $ stabilized by $ S $ has dimension $ dim(Q) = \frac{d^n}{ | S | } $. Thus, we demonstrate the following theorem.
\begin{theorem}\label{STA1}
Let $ S = \langle s_1, \ldots , s_r \rangle $ an abelian subgroup of the Pauli group $ \mathcal{G}_d^n $ where $ \{S_i \}_{ i = 1 }^{ r } $ are independent generators, which does not contain multiples of the  identity than the identity itself. Then the subspace stabilized by $ S $ has dimension $ \frac{d^n}{ | S | } $.
\end{theorem}

We will get now three important corollaries of the preceding theorem. The first two are used in the proof of Theorem~\ref{QuaCla}. The third result establishes the number of generators of $ S $ if $ d $ is prime.

\begin{corollary}\label{COLCARD}
Let $ S = \langle s_1 , \ldots , s_r \rangle $ be an abelian subgroup of the Pauli group $ \mathcal{G}_d^ n $ not containing multiple oh the identity than the identity itself. If $ | S | = d^n $ then $ S $ is a maximal set of Pauli operators that stabilizes a single state $ \dirac{\psi} $.
\end{corollary}

This corollary says that every Pauli operator $ P \in \mathcal{G}_{d}^{n} $ stabilizing $ \dirac{\psi} $ is in $ S $. The proof is given below.

\begin{proof}
By Theorem~\ref{STA1} we have $ S $ stabilizes a single state $\dirac{\psi}$. Suppose there is $ P\in \mathcal{G}_d^n $ and $ P \notin S $ that stabilizes $ \dirac{\psi} $. Clearly $P^t$ also stabilizes $ \dirac{\psi} $ for any $ t \in \mathbb{N} $, so there is no $ \overline{t} \in \mathbb{N} $ such that $ P^{ \overline{t}} = \alpha I $ with $ \alpha \neq $ 1. In addition, $ P $ commutes with all elements of $ S $ since otherwise, there would be $ s \in S $ and $ \beta \neq $ 1 such that $$\dirac{\psi}=P\dirac{\psi}=Ps\dirac{\psi}=\beta s P \dirac{\psi}=\beta \dirac{\psi}$$
which may not occur. Therefore, $ S = \langle s_1, \ldots , s_r , P \rangle $ is an abelian group not containing multiple of the identity than the identity with $ | S | > d^n $ and stabilizes $ \dirac{\psi} $, a contradiction to Theorem~\ref{STA1}.
\end{proof}

\begin{corollary}\label{maximal}
Under the same assumptions, unless phase, $ S $ is a maximal set of Abelian operators.
\end{corollary}
\begin{proof}
Suppose that $ P \in \mathcal{G}_d^n $ is a Pauli operator that commutes with all elements of $ S $. Then it follows that $ S $ stabilizes $ \dirac{ \psi } $ and $ P \dirac{\psi} $, but these vectors can not be linearly independent by Theorem~\ref{STA1}. Logo $ P \dirac{\psi } = \alpha \dirac{\psi} $, e.g, the operator $\alpha^{\dagger}$ stabilizes $ P \dirac{\psi} $. By the previous corollary, it follows that $ \alpha^{\dagger} P \in S $.
\end{proof}

\begin{corollary}
Let $ S = \langle s_1 , \ldots , s_r \rangle $ an abelian subgroup of the Pauli group $ \mathcal{G}_d^n $ not containing multiples of the identity than the identity itself and $ d $ prime. Then $ S $ stabilizes a single state $ \dirac{\psi} $ if and only if $ r = n $.
\end{corollary}
\begin{proof}
If $ d $ is prime, then the order of each generator is $o(s_i)=d \;\forall i$ , and it follows that $ | S | = d^r $. Then $ S $ stabilizes a single state if and only if $ r = n $.
\end{proof}

For $ d $ not prime, we may have a generator $ S_i $ of $ S $ with order less than $ d $, so that the required amount $ r $ of generators is greater than $ n $. The maximum number of generators is $ 2n $, as cited in \cite{hostens2005}, $ n \leq r \leq 2n $.

\section{CWS codes and Stabilizers codes}\label{CWSestabilizadores}

This section establishes relationships between CWS codes and stabilizer codes. There are several examples of codes that are not built with the CWS formalism, as we see in \cite{Grassl1997}, \cite{Grassl2008}, \cite{Grassl2008a} and \cite{Smolin2007}. There are also several CWS codes that are not stabilizers, how can we check in \cite{Chen2008}, \cite{Cross2009}, \cite{Yu2008}, \cite{Yu2007}, \cite{Yu2009}. Every stabilizer code is in fact a CWS code and all CWS code with \textit{codewords} $ W $ forming a group, is a stabilizer code. These results are shown in the binary case \cite{Cross2009} and for \textit{graph states} for any $ d $ \cite{Looi2007}, but we did not found in the literature a general statement, valid for any $ d $, and being not based on \textit{graph states}, so we did a demonstration based on the structure of the parity check matrix (Definition~\ref{MatrizV}). Given a set of Pauli operators $ \mathcal{C}$, the parity check matrix with coefficients in $ \mathbb{Z}_d $, $ R( \mathcal{C}) $. If the number of operators in $\mathcal{C} $ is $ l $, $ R(\mathcal{C}) $ represents a homomorphism between $ \mathbb{Z}_d $-modules, $ \mathbb{Z}_d^{2n} \rightarrow \mathbb{Z}_d^l $, so it makes sense to speak of kernel and image modules, respectively $ Ker(R(\mathcal{C})) $ and $ Im(R(\mathcal{C})) $. 

\begin{lemma}\label{lemaLinearisSta}
Let $ \mathcal{ Q } $ be a CWS code with stabilizer  $S$ generated by $ \mathbb{ S } = \{ s_1, \ldots, s_r \} $ and \textit{codewords} $ W = \{ w_i \}_{ i = 1 }^{ K } $. Then the cardinality of the centralizer of $ W $ in $ S $, $ C_S( W )$ and the cardinality $ \mathbb{ Z }_d $-module $ \langle R ( \mathbb{ S } ) \rangle \bigcap Ker(R(W)\Lambda ) $ are the same, e.g
$$ \# C_S ( W ) = \# \langle R ( \mathbb{ S } ) \rangle \bigcap Ker ( R ( W ) \Lambda ) $$
\end{lemma}
\begin{proof}
It suffices to show that the function
$$
\begin{array}{cccc}
f: & C_S(W) & \rightarrow & R(\mathbb{S})\rangle \bigcap Ker(R(W)\Lambda)\\
& g & \mapsto & R(g)
\end{array}
$$
is well defined, and is bijective.
\begin{enumerate}
\item $ f $ is well defined because if $ g_1\neq g_2 \in C_S ( W ) $, then $ R ( g_1^{ \dagger} g_2 ) \neq \mathbf{ 0 } $ and so $ R ( g_1 ) \neq R( g_2 ) $.
\item $ f $ is injective, because if $ R( g_1 ) = R ( g_2 ) $ then $ R ( g_1^{\dagger} g_2 ) = \mathbf{0} $, then $ g_1^{\dagger} g_2 = \alpha I $ and $ \alpha = 1$ because there is no multiple of the identity other than the identity itself in $ S $. 
\item $ f $ is surjective. Let $ \mathbf{v} \in \langle R(\mathbb{ S } ) \rangle \bigcap Ker( R ( W ) \Lambda ) $. There is $ \overline{ g } \in \mathcal{G}_d^n $ such that $ R( \overline{ g } ) = \mathbf{ v } $. As $ \mathbf{ v } \in \langle R( \mathbb{ S } ) \rangle $ and up to phase phase $ \langle R( \mathbb{ S } ) \rangle $ is a maximal abelian set, there is $ g = \alpha \overline{ g } $ with $ g \in S $ and $ R ( g ) = \mathbf{ v } $.

As $ \mathbf{ v } \in Ker( R ( W ) \Lambda ) $, $ R ( W ) \Lambda R^{T}( g ) = \mathbf{ 0 } $ it follows that $ g $ commutes with all elements of $ W $, then $ g \in C_S ( W ) $.
\end{enumerate}
\end{proof}

\begin{theorem}\label{LinearisSta}
Let $ \mathbb{ Q } $ be a CWS code with stabilizer  $S$ generated by $ \mathbb{ S } = \{ s_1, \ldots, s_r \} $. and codeword operators $ W = \{ w_i \}_{ i = 1 }^{ K } $ with $ w_1 =I $  \phantom{a} \footnote{ This condition is not restrictive since all CWS code is equivalent to a  $w_1 = I $. Just do  $ w_i'= w_{ 1 }^{ \dagger }w_i $.}. Then $ \mathbb{ Q } $ is a stabilizer code if and only if it satisfies $ \frac{ \# \langle R( W ) \rangle } { \# (\langle R( W ) \rangle \cap \langle R ( \mathbb{ S }) \rangle) } = K $.
\end{theorem}

\begin{proof}
Let $ \dirac{\psi} $ be the state stabilized by $ S $ and $ \beta = \{ w_i \dirac{ \psi } \}_{ i = 1 }^K $ a basis for $ \mathbb{ Q } $.  $ \mathbb{ Q }$ is a stabilizer code if and only if there exists a abelian subgroup $ H \leq \mathcal{G}_d^n $, not containing multiples of the identity than the identity itself that stabilizes $ \mathbb{ Q } $. In particular $ H $ need to stabilize $ \dirac{\psi } $. How $ S $ is a maximal subgroup that stabilizes $ \dirac{\psi} $ (Corollary~\ref{COLCARD}), then $ H \leq S $. Moreover, every element $ h \in H $ must satisfy $ hw_i= w_ih $ for all $ i $, then the subgroup $ H $ is the centralizer of $ W $ in $ S $, e.g $ H = C_S ( W ) $. It remains to show that $ \dfrac{d^n}{| C_S ( W ) | } = \frac{ \# \langle R ( W ) \rangle}{ \# \langle R ( W ) \rangle \cap \langle R ( \mathbb{ S } ) \rangle } $, so according to Theorem~\ref{STA1}, $C_S( W ) $ stabilizes $ \mathbb{Q} $ if and only if $ \frac{ \# \langle R ( W ) \rangle } { \# \langle R ( W ) \rangle \cap \langle R ( \mathbb{ S } ) \rangle } = K $.

According to Lemma~\ref{lemaLinearisSta}, ​​we have
$$ \# C_S ( W ) = \# \langle R ( \mathbb{ S } ) \rangle \bigcap Ker( R ( W ) \Lambda ) $$

Since $ S $ is up to phase a maximal abelian set in $ \mathcal{G}_d^n $ ( Corollary~\ref{maximal}), we have $ \langle R( \mathbb{ S } ) \rangle = Ker(R (\mathbb{ S } ) \Lambda) $, from which it follows that

\begin{eqnarray*}
 \langle R(\mathbb{S}) \rangle \bigcap Ker(R(W)\Lambda)&=& Ker(R(\mathbb{S})\Lambda) \bigcap Ker(R(W)\Lambda)\\
 & =& Ker(M) 
\end{eqnarray*}

where $M=\left[\begin{array}{c}
R(\mathbb{S})\Lambda \\ R(W)\Lambda
\end{array}\right]=\left[\begin{array}{c}
R(\mathbb{S}) \\ R(W)
\end{array}\right]\Lambda$. 
Then we estimate $ \# Ker ( M ) $.

 We have $ \langle M \rangle = \langle R ( \mathbb{ S } ) \Lambda \rangle + \langle R ( W ) \Lambda \rangle $. By the second isomorphism theorem for modules, we have
$$\frac{\langle R(\mathbb{S})\Lambda \rangle + \langle R(W)\Lambda \rangle}{R(\mathbb{S})\Lambda}\simeq \frac{\langle R(W)\Lambda \rangle}{\langle R(W)\Lambda \rangle \cap \langle R(\mathbb{S})\Lambda \rangle},$$
  and how the operator $ \Lambda $ does not change the cardinality of the row-module, we have $ \# \langle M \rangle = \frac{ \# \langle R ( \mathbb{ S } ) \rangle \# \langle R ( W ) \rangle } { \# \langle R ( W ) \rangle \cap \langle R ( \mathbb{ S } ) \rangle } $ and therefore as $ \frac{ \# \langle R ( W ) \rangle }{ \# \langle R ( W ) \rangle \cap \langle R ( \mathbb{ S } ) \rangle } = K $ and $ \# \langle R ( \mathbb{ S } ) \rangle = | S | = d^n $, we have $ \# \langle M \rangle = Kd^n $.
As seen, the cardinality of the row-module is equal to the cardinality of the column-module, so $ \# Im( M ) = Kd^n $. Finally, by the first isomorphism theorem, we have $\#Ker(M)=\dfrac{d^{2n}}{Kd^n}=\frac{d^n}{K}$.
\end{proof}

{\bf Example}

Take the $((3,3,2))_3$ code  with stabilizer $S = \langle s_1, s_2, s_3 \rangle $ where $ s_1 = XZI $, $ s_2 = ZXZ $ and  $s_3=IZX $ and \textit{codewords} $ W = \{ I , ( XZ ) \otimes Z \otimes Z^2,(XZ^2) \otimes Z \otimes Z \} $. We have respectively:
\begin{align*}
R(\mathbb{S})=\left[\begin{array}{cccccc}
0 & 1 &0 &1& 0& 0 \\ 
1 &0 &1 &0 &1 &0 \\
0 &1 &0 &0 &0 &1 
\end{array}\right]
\end{align*}
e
\begin{align*}
R(W)=\left[\begin{array}{cccccc}
0 & 0 &0 &0& 0& 0 \\
1 &1 &2 &1 &0 &0 \\
2 &1 &1 &1 &0 &0 
\end{array}\right]
\end{align*}
the row-module, $ \langle R ( W ) \rangle $ is represented by the following vectors:
$$\begin{array}{|c|c|}\hline
000000 & 112100\\\hline
010100 & 211100\\\hline
020200 & 221200\\\hline
       & 122200\\\hline
       & 201000\\\hline
       & 102000\\\hline
\end{array}$$
where the left are those belonging to $ \langle R ( \mathbb{ S } ) \rangle \cap \langle R ( W ) \rangle $. Then we see that $ \frac { \# \langle R ( W ) \rangle }{ \# \langle R ( W ) \rangle \cap \langle R ( \mathbb{ S } ) \rangle } = K $, then the code is stabilizer by Theorem~\ref{LinearisSta}. Actually, we can see that the code is equivalent to the code $ [ [ 3,1,2 ] ] _3 $ in \cite{Hu2008} with stabilizer $S'=\langle ZXZ, XZ^2X \rangle$.

From Theorem~\ref{LinearisSta}, follows two Corollaries representing results usually found in the literature.
\begin{corollary}\label{col1}
Let $ \mathcal{ Q } $ be a CWS code with stabilizer $ S = \langle s_1, \ldots, s_r \rangle $ and \textit{codewords} $ W = \{ w_i \}_{ i = 1 }^{ K } $ forming a group. Then $ \mathbb{ Q } $ is a stabilizer code.
\end{corollary}
\begin{proof}
If $ W $ is a group, then the rows of $ R ( W ) $ also form an additive group, then $ \# \langle R ( W ) \rangle = \# W = K $. Moreover, we have by the construction of CWS codes that, $ \langle R ( W ) \rangle \cap \langle R ( \mathbb{ S } ) \rangle = \{ 0 \} $, so $ \frac{ \# \langle R ( W ) \rangle }{ \# \langle R ( W ) \rangle \cap \langle R ( \mathbb{ S } ) \rangle } = K $
\end{proof}

\begin{corollary}\label{col2}
Let $ \mathbb{ Q } $ be a CWS code with stabilizer $ S = \langle s_1, \ldots, s_r \rangle $, \textit{codewords} $ W = \{ w_i \}_{ i = 1 }^{ K } $ with $ w_1 = I $ and stabilized state $ \dirac{\psi}$. If the classic words $ Cl_S( W ) $ form a group, then the code is a stabilizer one.
\end{corollary}
\begin{proof}
To show that $ \frac{ \# \langle R ( W ) \rangle }{ \# \langle R ( W ) \rangle \cap \langle R ( \mathbb{S} ) \rangle } = K $ is enough to show that every element $ r_w \in \langle R ( W ) \rangle $ is of the form $ r_w = R ( w_j ) + r_s$ with $ r_s \in \langle R ( \mathbb{ S } ) \rangle $. The transformation $Cl_S$ has domain in $ \mathcal{G}_d^n $. each element of $ \mathcal{G}_d^n $ has a representation on $ \mathbb{Z}_{d}^{2n} $. As already seen (equation~\ref{funcaocls}), we can describe the transformation  $Cl_S $ over $ \mathbb{ Z }_{ d }^{ 2n } $ as a homomorphism of modules represented by the matrix $ \mathcal{T} = R ( \mathbb{S} ) \Lambda $.

Take then $ r_w \in \langle R ( W ) \rangle $, so $r_w = \alpha_1 R ( w_1 ) + \ldots + \alpha_kR( w_k ) $ and
\begin{align*}
\mathcal{T}(r_w)&=\alpha_1\mathcal{T}(R(w_i))+\ldots +\alpha_k\mathcal{T}(R(w_k))\\
&=\alpha_1c_1+\ldots +\alpha_kc_k.
\end{align*}
As $Cl_S(W) $ form a group, the last summation is $\mathcal{T}(R(w_j))=c_j\in Cl_S(W)$, e.g $ \mathcal{T}( r_w ) = \mathcal{T}( R ( w_j ) ) $ so
$$ r_w = R ( w_j ) + r_s $$
where $ r_s \in Ker( \mathcal{T} ) = \langle R(\mathbb{S} ) \rangle $.
\end{proof}

It also follows that any stabilizer code can be seen as a CWS code, as shown in the following theorem
\begin{theorem}
All stabilizer code  $ \mathcal{Q} $ is a CWS code.
\end{theorem}

\begin{proof}
Let $ S' = \langle s_1, \ldots, s_m \rangle $ be the stabilizer group of the code $ \mathcal{ Q } $  and let $ dim(\mathcal{ Q } ) = K $. As already discussed, $ S'$ can then be extended to a maximal group $ S = \langle s_1, \ldots, s_m, g_1, \ldots, g_L \rangle $ with cardinality $ | S | = d^n $. This group stabilizes a single state $ \dirac{\psi} \in \mathcal{ Q } $. Consider now the parity check matrix $ R ( \mathbb{S}) $. We have $ \# \langle R ( \mathbb{S} ) \rangle = d^n $. As the cardinality of the row-module is the same of the column-module, we have $ \# Im( R ( \mathbb{S} ) ) = d^n $ and in turn also $ \# Im(R(\mathbb{ S } ) \Lambda ) = d^n $. This equality implies that for every $ \mathbf{ x } \in Im( R ( \mathbb{ S }) \Lambda ) $, there is a Pauli operator $ P_\mathbf{ x } $ such that $ \mathcal{ H }_{ d }^{ n } = \bigoplus P_\mathbf{ x } \dirac{\psi} $ and each state $ P_\mathbf{ x } \dirac{\psi} $ is the intersection of the eigenspaces associated with eigenvalues $q_d^{x_i}$ for each generator of $ S $. Since $ \mathcal{ Q } $ is a stabilizer code and $ dim ( \mathcal{ Q } ) = K $   we know that there are $K$ of these Pauli operators forming a set $ W = \{ P_{ x_i } \}_{ i = 1 }^{ K } $ that form a basis for $ \mathcal{ Q } $. Then just take the set $ W $ as \textit{codewords}
\end{proof}

\section{Conclusion}

In Section~\ref{structure}, was demonstrated for qudits, that a stabilizer group $ S $ of order $ | S | $ stabilizes a subspace of $ \mathcal{ H }_d^n$ of dimension $ \frac{d^n}{ | S | } $. Although there is already a demonstration of this result, we created a proof which generalizes the one for qubits contained in \cite{Nielsen2000} and makes use of the parity check matrix of Definition~\ref{MatrizV} and the interpretation of the matrix $ R ( \mathbb{ S } ) \Lambda $ as a homomorphism of $ \mathbb{ Z }_d $-modules.

In Section~\ref{CWSestabilizadores} we use the parity check matrix $ R ( \mathbb{S} ) $ and the interpretation of the matrix $ R ( \mathbb{ S } ) \Lambda $ as a homomorphism from $ \mathbb{ Z }_d $-modules to prove Theorem~\ref{LinearisSta} which generalizes the results contained in Corollaries~(\ref{col1} and \ref{col2}). These corollaries are accepted results in the literature, but hard to find for qudits.


\section*{Acknowledgment}

The authors would like to thank FAPEMIG.



%

\bibliographystyle{IEEEtran}
\bibliography{quantum-codes}



\end{document}